\documentclass[12pt,a4paper]{article}
\usepackage[T1]{fontenc}
\usepackage[utf8]{inputenc}
\usepackage{mathtools,amsmath,amsthm,array,graphicx,natbib,color,xspace,enumitem,multirow}
\usepackage[textsize=tiny]{todonotes}
\usepackage{pgfplots}
\usetikzlibrary{decorations.pathreplacing}

\newtheorem{mydef}{Definition}
\newtheorem{prop}{Proposition}

\newtheorem{example}{Example}
\newtheorem{lemma}{Lemma}

\usepackage{fancyhdr} 
\usepackage{todonotes}

\begin{document}

\begin{center}\huge{Observing Actions in Bayesian Games}
\end{center}

\begin{center} \textit{Dominik Grafenhofer, Wolfgang Kuhle}\footnote{University of Economics, Prague, Czech Republic. Max Planck Institute for Social Law and Social Policy, Munich, Germany, E-mail wkuhle@gmx.de. Dominik Grafenhofer, E-mail econ@grafenhofer.at. Most of this paper was written during our time at the Max Planck Institute in Bonn, and we thank Martin Hellwig and Carl Christian von Weizsäcker for helpful and encouraging conversations. We also thank Philipp Koenig and Seminar participants in Bonn, Prague and Hangzhou for comments and questions.}
\end{center}

\noindent\emph{\textbf{Abstract:} We study Bayesian coordination games where agents receive noisy
private information over the game's payoff structure, and over
each others' actions. If private information over actions is
precise, we find that agents can coordinate on multiple
equilibria. If private information over actions is of low quality,
equilibrium uniqueness obtains like in a standard global games
setting. The current model, with its flexible information structure, can thus be used to study phenomena such as bank-runs, currency crises, recessions, riots, and revolutions, where agents rely on information over each others' actions.}


\noindent\emph{ERNIE: There is something funny going on over there at the
bank, George, I've never really seen one, but that's got all the
earmarks of a run.}\\
\emph{PASSERBY: Hey, Ernie, if you have any money in the bank, you
better hurry.}\footnote{From the movie ``It's a wonderful life."}


\section{Introduction}

Coordination games are used extensively to model situations, such as bank runs, currency crises, or riots and revolutions. The distinguishing feature of such games is that agents have a strong incentive to choose mutually consistent strategies.

One way to study such games is to assume complete information,
i.e., players know the model's payoff relevant coefficients and
each others' equilibrium actions. This approach tends to
produce multiple, pure strategy, equilibria. These equilibria have
been criticized on the grounds that players often cannot observe
the game's payoff structure with perfect accuracy. \citet{Rub89},
\citet{Car93} and \citet{Mor98} argue that agents have to rely on
noisy private information over the game's payoff relevant
coefficients.\footnote{\citet{Har88} review the
earlier literature on strategic uncertainty and equilibrium
selection.} In such games, each agent has to use his own signal over the game's
payoffs to infer the signals, and thus the actions, of the other
agents. Such inference makes it difficult to coordinate on
multiple equilibria. In particular, whenever private information over
the game's payoffs is very precise, but not perfect, the global
games structure ensures equilibrium uniqueness.

To understand the equilibria that agents play, we argue that agent's information over the game's payoff coefficients \textit{and} information over each other's actions are of particular importance. The global games approach of \citet{Rub89}, \citet{Car93} and \citet{Mor98}, focuses only on one of these types of information, i.e. information over fundamentals, and assumes away the other. In the current paper, we build a model where agents simultaneously use information on the game's payoffs, as well as, information over each others' actions. 

Information over actions plays a dual role. First, actions depend on the game's payoff coefficient, and signals over actions therefore carry information over the games' payoff coefficients. In this interpretation, the signal over actions is just another signal over the game's fundamental. Taking this view, private information over actions should reinforce the global games equilibrium selection mechanism, where private information over fundamentals generates unique equilibria. The second function of signals over actions is, of course, that they inform players of each others' actions, which helps coordination. In equilibrium, we find that this second effect dominates, and multiple equilibria are ensured whenever private information over actions is sufficiently precise. 

In our model, runs feed on themselves: the size of the attack $A$ is an increasing function $\pi(A)$ of the attack itself. That is, increases in the mass of attacking agents are observed, and induce additional agents to join the run. In turn, this increase in attacking agents convinces even more bystanders to join the run. The strength of this infectious process increases with the private signal's precision, and agents can coordinate on multiple equilibria whenever the precision with which they observe each other is sufficiently high.

\emph{Interpretation:} Riots and revolutions, currency crises, bank-runs, recessions, or flights to quality, are social phenomena where agents observe, learn from, and emulate each others' actions. That is, in the context of a bank-run, depositors can observe the length of a queue,
respectively the lack thereof, in front of their local bank
branch.\footnote{Indeed, to avoid queues, which extend out to the
street, bank lobbies are traditionally rather large. That is, large lobbies reduce the precision with which agents can observe each other, which makes coordination harder. \citet{Dia83}, p. 408, stress the importance of queues in the context of bank runs. Their "sequential service constraint" formalizes that agents who withdraw early can front-run agents who come late. The current model emphasizes that such a delay in service is observable, and attracts agents who wouldn't withdraw otherwise.}
Similarly, during a currency crises, commercial banks observe their clients'
order flow, which helps them decide whether or not it is worth
while to join the run. Individual traders, who participate in the currency market, have conversations about their positions in the "Mexican Peso" or the "British Pound." Cartel members examine whether their
partners are undercutting the agreed on price. In the context of
business-cycles, where firms have an incentive to
produce whenever the other's are producing, firms closely monitor the level of "economic activity".\footnote{See the \citet{Dia82} type models for macro
settings, where agents' production choices are strategic complements.}
Finally, during riots and revolutions, the inhabitants of large
towns can see whether the number of protesters in the street is
large or small. In turn, "policy makers" enforce curfews, limit internet access, and shut down social media sites in an effort to conceal unrest.

\emph{Related literature:} \citet{Rub89}, \citet{Car93}, \citet{Mor98},
\citet{Fra04} show that global games select unique equilibria when private information
over the game's fundamental is sufficiently precise.
\citet{Mor01,Mor04}, \citet{Hel02} and \citet{Met02} emphasize that
public signals can restore multiplicity. \citet{Atk01} and \citet{Ang06} argue that stock prices aggregate and publicize private information in a manner that brings back multiplicity when private information over fundamentals is sufficiently precise. \citet{Hel02}, \citet{Mor07} and \citet{Mon89} show that the global games mechanism relies on signal structures that generate low common p-belief.

The current model links the global games theory to the herding literature, \citet{Sch90}, \citet{Bik92} and \citet{Ban92}, where agents' actions reveal information over unknown fundamentals. This channel is present in the current model. Private information over actions informs agents about the game's fundamental. Following the classic global games logic, such a signal should lower common p-believes, respectively, make it harder for agents to coordinate on multiple equilibria. While present, this channel is overcompensated by the fact that agents also learn about each others' strategies, which makes coordination easier.         

\citet{Ang06} and \citet{Das07} find that public signals
over actions, just like public signals over fundamentals, help agents to coordinate on multiple equilibria. The present model shows that private signals over actions,
unlike private signals over fundamentals, induce multiple
equilibria if they are sufficiently precise. That is, contrary to
the global games logic, where precise ``private information"
yields unique equilibria, we provide an example where private
signals ensure multiple equilibria.

The importance of information over actions has been emphasized in
a separate literature on ``conjectural equilibria." \citet{Bat97},
\citet{Min03}, \citet{Rub94}, and \citet{Esp13}, develop models where agents
receive noisy signals over each other's actions.\footnote{In a
similar gist, \citet{Hah77,Hah78} analyzes Walrasian economies,
where agents hold conjectures over each others' supply and demand
functions, which need not be true. \citet{Gue02} studies the relation between common knowledge, agents' conjectures, and the theory of rational expectations.} The current model thus brings
together arguments from the literature on conjectural equilibrium,
emphasizing noisy information over actions, with arguments from
the literature on global games, where uncertainty over actions
originates from parameter uncertainty.

Finally, we contribute to the effort aimed at
enriching the global games structure. \citet{Izm10},
\citet{Ste11}, \citet{Kuh15}, \citet{Gra16}, \citet{Ber16},
\citet{Bin01}, \citet{Pav07}, \citet{Mat12}, \citet{Fra12} add
heterogenous priors, multidimensional signal structures,
applications to mechanism design and dynamic information
revelation.

\emph{Organization:} Section \ref{model} outlines the model. Sections \ref{information_equilibrium} and \ref{generalerrorterms} contain the main results. In Section \ref{Tipping}, we reduce the dimensionality of our model, and Section \ref{Normal} illustrates our findings for normally distributed variables. Section \ref{Discussion} concludes.


\section{A coordination game with dominance regions} \label{model}

There is a status quo and a unit measure of agents indexed by
$i\in[0,1]$. Each of these agents $i$ can choose between two
actions $a_i\in\{0,1\}$. Choosing $a_i=1$ means to attack the
status quo. Choosing $a_i=0$ means that the agent abstains from
attacking the status quo. An attack on the status quo is
associated with a cost $c\in(0,1)$. If the attack is successful,
the status quo is abandoned, and attacking agents receive a net
payoff $1-c>0$. If the attack is not successful, an attacking
agent's net payoff is $-c$. The payoff for an agent who does not
attack is normalized to zero. The status quo is abandoned if the
aggregate size of the attack $A:=\int_0^1a_idi$ exceeds the
strength of the status quo $\theta$, i.e., if $A>\theta$.
Otherwise, if $A<\theta$, the status quo is maintained, and the
attack fails.

\begin{center}
\begin{tikzpicture}[xscale=1]
\draw [<->](0,0) -- (9,0);
\draw (3,-.2) -- (3, .2);
\draw (6,-.2) -- (6, .2);
\node [left] at (10,0) {$\theta$};
\node[align=center, above] at (3,.3){$0$};
\node[align=center, above] at (6,.3){$1$};
\node[align=center, below] at (1.5,-.3){Attack};
\node[align=center, below] at (4.5,-.3){Coordination};
\node[align=center, below] at (7.5,-.3){No attack};
\end{tikzpicture}\label{p12} 
\end{center}
\begin{center} Figure \ref{p12} \end{center}

Figure \ref{p12} illustrates that the fundamental $\theta$ can fall into three regions. First, if $\theta<0$, the status quo is abandoned, regardless of how many agents attack. If an agent knows that $\theta<0$, he should attack regardless of strategic considerations. Second, when $\theta\in(0,1)$, agents are in a coordination game, where the mass of attacking agents is crucial for the game's payoffs. Third, if $\theta>1$, the mass of attacking agents is always too small to overcome the status quo, and not attacking is the dominant strategy. To characterize the equilibria that agents play when they have incomplete information over the fundamental $\theta$, we start with an example. 

\section{Information and Equilibrium}\label{information_equilibrium}

Regarding agents' information we assume that players hold a uniform uninformative prior over the distribution of $\theta$. Moreover, we assume that each player $i$ receives two kinds of private information: a noisy signal $x_i$ over the strength of the status quo $\theta$, and another signal $y_i$ over the other players' actions $A$:
\begin{equation} \label{def_signals_x_and_y}
\begin{array}{lrcl}
\text{Signal over the fundamental: } \qquad\qquad & x_i &=& \theta + \epsilon^x_i \\
\text{Signal over the aggregate attack: } \qquad\qquad &  y_i &=& A + \epsilon^y_i
\end{array}
\end{equation}
The structure of the game, in particular the distribution of error terms in (\ref{def_signals_x_and_y}), is common knowledge. The exogenous fundamental $\theta$ and the endogenous size of the attack $A(\theta)$, however, are not common knowledge. Each agent uses his private signals $x_i,y_i$ to infer these two key variables.

Agents choose their strategies $a_i$ to maximize expected utility: 
\begin{eqnarray} E[U(a_i)|x_i,y_i]=a_i(P(\theta<A|x_i,y_i)-c). \end{eqnarray}
Action $a_i=1$ is thus optimal whenever $P(\theta<A|x_i,y_i)\geq c$. Agents are just indifferent between attacking and not attacking when signals $x_i$ and $y_i$ are such that $P^*:=c$. Finally, we denote the joint probability distribution function of the error terms by $f(\epsilon^x,\epsilon^y)$. Given agents' information, and the critical probability $P^*$, we can define:  
\begin{mydef}[Equilibrium] \label{fe}
An aggregate attack function $A(\theta)$ is an equilibrium of the game, if for all $\theta\in\mathbf{R}$ the following holds\footnote{$\chi$ denotes the indicator function.}:
\begin{equation}\label{eq_consistency}
A(\theta) = \int_{\mathbf{R}^2} \chi_{\{(\epsilon^x,\epsilon^y):P[A(\theta)\geq\theta|x_i=\theta+\epsilon^x,y_i=A(\theta)+\epsilon^y,A(\cdot)]\geq P^*\}} \, f(\epsilon^x,\epsilon^y) \, d\epsilon^x d\epsilon^y.
\end{equation}
\end{mydef}

To characterize equilibria, we start with an example where signal errors are bounded:
\begin{prop}\label{cs_mult}
Suppose that $\epsilon^y_i\in[-\sigma,\sigma]$, i.e. $f(\epsilon^x,\epsilon^y)=0$ for all $|\epsilon^y|>\sigma$ and all $\epsilon^x\in\mathbf{R}$. Further assume that the precision of the signal about the aggregate attack $y_i$ is precise enough, i.e. $0<\sigma<\frac12$. Then, there exists a continuum of equilibria.
\end{prop}
\begin{proof}
Pick an arbitrary $t\in[0,1]$ and define
\begin{equation}\label{Attheta}
A_t(\theta) := \left\{\begin{array}{ll}
1 \qquad\qquad & \theta < t \\
0 \qquad\qquad & \text{otherwise.}
\end{array} \right.
\end{equation}
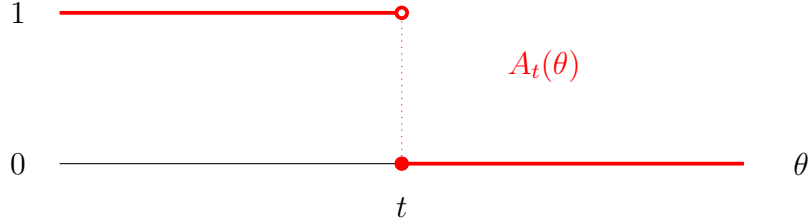
\begin{center}
	\begin{tikzpicture}[xscale=1]
	\draw (0,0) -- (9,0);
	\draw [line width=0.5mm,color=red](0,2) -- (4.45,2);
	\draw [line width=0.5mm,color=red](4.5,2) circle (2pt);
	\draw [line width=0.5mm,color=red](4.55,0) -- (9,0);
	\filldraw [line width=0.5mm,color=red](4.5,0) circle (2pt);
	\draw [dotted, color=red](4.5,.1) -- (4.5, 1.9);
	\node [left] at (10,0) {$\theta$};
	\node[align=center, below] at (4.5,-.3){$t$};
	\node[align=center, left] at (-.3,0){$0$};
	\node[align=center, left] at (-.3,2){$1$};
	\node[align=center, left,color=red] at (7,1.3){$A_t(\theta)$};
	\end{tikzpicture}\label{Diagram2}
\end{center}
To establish that $A_t$ functions (\ref{Attheta}) are equilibria, we have to show that \eqref{eq_consistency} holds. First, we determine the distribution of signal realizations $y_i$ for a given $\theta$. There are two cases:
\begin{equation}\label{cutoff}
\begin{array}{ll}
\text{Case } \theta<t: & \qquad y_i = 1 + \epsilon^y_i > 1-\sigma \geq \frac12\\
\text{Case } \theta\geq t: & \qquad y_i = 0 + \epsilon^y_i < \sigma \leq \frac12.
\end{array}
\end{equation}
Notice, that signal realizations $y_i$ do not overlap: when $\theta<t$ all signal realizations are above $\frac12$. On the contrary, whenever $\theta>t$ all signal realizations are below $\frac12$.
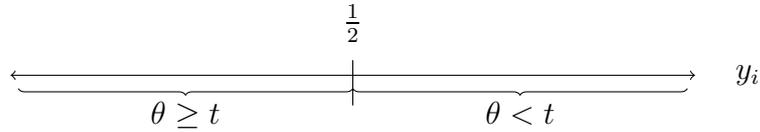
\begin{center}
	\begin{tikzpicture}[xscale=1]
	\draw [<->](0,0) -- (9,0);
	\draw (4.5,-.4) -- (4.5, .2);
	\node [left] at (10,0) {$y_i$};
	\node[align=center, above] at (4.5,.3){$\frac12$};
	\draw[decoration={brace,mirror,raise=5pt},decorate]
	(0.1,0) -- node[below=6pt] {$\theta\geq t$} (4.5,0);
	\draw[decoration={brace,mirror,raise=5pt},decorate]
(4.5,0) -- node[below=6pt] {$\theta<t$} (8.9,0);
	\end{tikzpicture}\label{Diagram3}
\end{center}
Now suppose that an agents learns that his signal $y_i$ is larger than $\frac12$: then he knows that this is only possible when $\theta<t$. In turn, $\theta<t$ means that a successful attack is underway, which he should join. Indeed, using the conjecture $A_t(\theta)$, we know that all agents will attack.

If the agent's signal $y_i$ is smaller than $\frac12$, he knows $\theta<t$. That is, given the conjectured aggregate attack $A_t(\theta)$, the attack is not successful. Hence, it is optimal for the agent to abstain from attacking.

Equation \ref{eq_consistency} holds for all $\theta$: when $\theta<t$ all agents attack, and the aggregate attack is 1. In the other case, no agent attacks and the aggregate attack is 0. Thus, $A_t(\theta)$ satisfies the requirements of an equilibrium.
\end{proof}

Within the current class of equilibria (\ref{Attheta}), coordination is perfect: Diagram \ref{Diagram2} illustrates that we have either a successful attack $A_t=1$, in which all agents participate, or no attack $A=0$. This perfect coordination relies on the compact support of $y_i|A$, i.e, the error term. Large realizations of $A$ shift $y_i$ signals such that all agents are willing to attack. Small realizations of $A$ shift agents' observations $y_i$ such that all agents abstain from attacking. In Section \ref{generalerrorterms}, we relax the assumption of compact support and derive our results for games where the noise terms' support is unbounded. For these games, signals over actions are never fully revealing, but they are still sufficiently informative such that  runs feed on themselves, and multiple equilibria obtain. That is, increases in $A$ shift observations $y_i$ by enough such that a sufficient number of agents join a run. Except for the noise in agents' private signals, we maintain the assumptions from the previous section. 



\subsection{Unbounded Errors}\label{generalerrorterms}


When agents' signal errors have unbounded support, agents can never be sure entirely sure whether they are joining a successful run or a hopeless effort. Hence, contrary to Proposition \ref{cs_mult}, perfect coordination is impossible. That said, runs continue to feed on themselves in the same manner as before, and multiple equilibria obtain whenever private information over actions is sufficiently precise: 

\begin{prop}\label{infogap_mult_general}
Assume that $\epsilon^x_i$ and $\epsilon^y_i$ are distributed according to pdf $f=f_x f_y$ (cdfs $F_x$, $F_y$), where $f_x$ and $f_y$ are symmetric.\footnote{The symmetry assumption shortens the proof, and can otherwise be discarded.} There exist $\delta> 0, \gamma> 0,$ and $\xi > 0$ such that $1-\delta \geq \gamma$, $1>3\delta + 2\gamma$ and the following conditions hold:
\begin{align}
&\frac{F_x(\xi)}{1-F_x(\xi)} \, \frac{\sup_{a\in[0,\delta]}f_y(\eta-a)}{\inf_{a\in [1-\delta,1]}f_y(\eta-a)} \leq \frac{1-c}{c} \quad \text{ for all }\eta\geq 1-\delta-\gamma \label{cond1_2d}\\
&F_x[\xi]F_y[\gamma]\geq 1- \delta \label{cond2_2d} 
\end{align}
Whenever these conditions hold, there exists a continuum of equilibria.
\end{prop}
\begin{proof} See Appendix \ref{proof_infogap_mult_general}.\end{proof}
The proof of Proposition \ref{infogap_mult_general} relies on an iteration argument: we start with a guess $A^0_t(\theta)$ and compute a best response $A^1_t(\theta)$. In turn, $A^1_t(\theta)$ yields another best response $A^2_t(\theta)$ and so on... . In a second step, we show that this iteration allows to construct a converging sequence of aggregate attacks/best responses. In the limit, we obtain an equilibrium attack function for every $t$. Finally, since there is a continuum of permissible $t$ values, we have a continuum of equilibria. 

To interpret conditions (\ref{cond1_2d}) and (\ref{cond2_2d}), we use an 

\begin{example}\label{example}
Suppose error terms $\epsilon^x,\epsilon^y$ are normally distributed, and denote the precisions of the respective signals by $\alpha_x=\frac{1}{\sigma^2_x}$ and $\alpha_y=\frac{1}{\sigma^2_y}$. For that case, we have $f_x(\xi)=\phi(\sqrt{\alpha_x}\xi)$ and $F_x(\xi)=\Phi(\sqrt{\alpha_x}\xi)$ and $f_y(\gamma)=\phi(\sqrt{\alpha_y}\gamma)$ and $F_y(\gamma)=\Phi(\sqrt{\alpha_y}\gamma)$, where $\phi$ and $\Phi$ represent the density and cumulative density functions of the standard normal distribution. In turn, we can choose $\delta$ and $\gamma$ such that $1>3\delta+2\gamma$, e.g., $\delta=.2$ $\gamma=.1$. Moreover, we can choose $\alpha_x\xi>\Phi^{-1}(.8)$. Finally we let $\alpha_y\rightarrow\infty$, such that (\ref{cond1_2d}) and (\ref{cond2_2d}) are both satisfied.\end{example}

Example \ref{example} allows for an interpretation of our findings: multiple equilibria are ensured only when the precision with which agents can observe actions is large. To illustrate the mechanisms behind the multiplicity result in Proposition \ref{infogap_mult_general} more clearly, we now focus on an alternative, one-dimensional, signal structure.

\section{Information over the Tipping Point}\label{Tipping}

Instead of observing fundamentals and actions separately, we now think of a model where agents simply observe one signal that informs them whether the economy is far away or close to the "tipping point." That is, agents observe a signal over $A-\theta$, which informs them whether the status quo will be abandoned or not. In the context of a currency crisis, such a variable can be interpreted as net outflows from a specific currency, respectively, the change in foreign currency reserves of the central bank. In the context of riots and revolutions, agents observe the degree to which protesters and outnumber the police.
We proceed in two steps. First, we study general distribution functions, and distinguish between bounded and unbounded error terms. Compared to the results from Sections \ref{information_equilibrium}, we find that "precise private information" always ensures multiple equilibria. Second, we add an example, where errors are normally distributed, as an illustration. This example makes transparent how multiple run equilibria obtain in situations, where agents observe each others actions.  

Suppose that agents receive only one signal:
\begin{equation}\label{infogap}
z_i = A-\theta + \rho_i,
\end{equation}
which informs about the likelihood that an attack succeeds. 

\begin{prop}\label{infogap_mult}
	Suppose that $\rho_i\in[-\sigma,\sigma]$, and that the precision of the signal about the aggregate attack $y_i$ is precise enough, i.e. $0<\sigma<\frac12$ Then, there exist a continuum of equilibria.
\end{prop}

\begin{proof}[Sketch of proof]
The proof is parallel to that of Proposition \ref{cs_mult}: pick an arbitrary $t\in(\sigma,1-\sigma)$ and define $A_t(\theta)$. The signal realizations for different $\theta$ are:
\begin{equation}\label{infogap_cutoff}
\begin{array}{ll}
\text{Case } \theta<t: & \qquad z_i = 1-\theta + \rho_i > \sigma + \rho_i \geq 0 \\
\text{Case } \theta\geq t: & \qquad z_i = -\theta + \rho_i \leq -\sigma + \rho_i \leq 0.
\end{array}
\end{equation}
Again, there is no overlap in signal realizations in both cases, and the remainder of the argument is parallel to the proof of Proposition \ref{cs_mult}.
\end{proof}

\subsection{Unbounded Errors}

We denote by $G$ the cumulative distribution function of the error $\rho_i$, and by $g=G'$ the respective probability distribution function\footnote{We assume that $G$ is differentiable, i.e., we rule out atoms}. Again, we simply our proof and assume that $g$ is supposed to be symmetric around $0$. Regarding $\theta$, we assume that agents hold an uninformative uniform prior.

\begin{prop}\label{infogap_mult_1d}
Suppose that there exist $\delta>0$ and $\gamma>0$ such that the following conditions hold:
\begin{align}
&\frac{1-G(\xi-\alpha)}{G(\xi-\beta)} \geq \frac{1-c}{c}, \quad\text{for all }\xi\geq 1-\delta-\gamma,\, \alpha\in[0,\delta],\text{ and } \beta\in[1-\delta,1] \label{cond1_infogap_mult_1d}\\
& G(\gamma) \geq 1 - \delta \text{, and} \label{cond2_infogap_mult_1d}\\
& g(\delta-\gamma) < 1\label{cond3_infogap_mult_1d}
\end{align}
Then, there exists a continuum of equilibria.
\end{prop}
\begin{proof}
See Appendix \ref{proof_infogap_mult_1d}.
\end{proof}

The proof strategy is to start for fixed $t$ with a guess of the equilibrium aggregate attack, and then compute the best response to that attack. This yields another guess for the aggregate attack, to which we compute the best response.... . We use this argument to construct a converging sequence of aggregate attack guesses. The limit aggregate attack is an equilibrium, which is different for different choices of $t$.

\subsection{Normally Distributed Errors}\label{Normal}
To illustrate the origins of equilibrium multiplicity, let us assume that agents $i\in{[0,1]}$ receive private information:
\begin{eqnarray}z_i=A-\theta+\sigma_z\xi_i, \quad \xi_i\sim\mathcal{N}(0,1).\label{i1}\end{eqnarray}
The signal $z_i$ therefore informs player $i$ with
precision $\alpha_z:=\frac{1}{\sigma_z}$ of the attack's net size
$A-\theta$, and we have
\begin{prop}
If private information is precise, $\alpha_z>\sqrt{2\pi}$, agents
can coordinate on multiple equilibria.\end{prop}
\begin{proof}
We proceed in three steps. First, we compute the threshold signal
$z^*$, for which agents are indifferent between attacking and not
attacking. Given this threshold, we compute the mass of attacking
agents. Third, we show that there exist multiple equilibria.

1.) Payoff indifference condition (PIC): Given a signal $z_i$,
agents $i$ choose an action $a_i\in\{0,1\}$, to maximize expected
utility:
\begin{eqnarray} E[U_i]=a_i(P(A-\theta>0|z_i)-c). \end{eqnarray}
Agent $i$ is therefore just indifferent between attacking,
$a_i=1$, and not attacking, $a_i=0$, when he receives a signal
$z_i=z^*$ such that:
\begin{eqnarray}  P(A-\theta>0|z^*)=c. \label{l0.1} \end{eqnarray}

It follows from (\ref{l0.1}) that agents attack if $z>z^*$, and
they will abstain from attacking whenever $z\leq z^*$.

2.) Given the critical signal $z^*$, we can compute the mass of
attacking agents:
\begin{eqnarray} A=P(z>z^*|A,\theta)  \label{A1}\end{eqnarray}
For normally distributed signal errors, (\ref{A1}) can be
rewritten as:
\begin{eqnarray} A=1-\Phi(\alpha_z(z^*-A+\theta)), \label{A2}  \end{eqnarray}
where $\Phi()$ is the cumulative normal distribution. From
(\ref{A2}), we have 
\begin{lemma} \label{lem1} If $\alpha_z>\frac{1}{\sqrt{2\pi}}$ then, for every level $z^*$, there exists an
interval $[\check{\theta}(z^*),\hat{\theta}(z^*)]$ such that
(\ref{A2}) has three solutions $A_j(\theta,z^*), j=1,2,3$ whenever
$\theta\in[\check{\theta}(z^*),\hat{\theta}(z^*)]$.
\end{lemma}

To construct equilibrium functions $A_t(\theta,z^*)$, we use the
solutions $A_j(\theta;z^*)$ from Lemma \ref{lem1}. More
specifically, we focus on the solutions $j=1,3$ such that we
obtain functions $A_t(\theta;z^*)$, which are downward sloping
$\frac{\partial A_t}{\partial \theta}<0$. Given these functions,
it remains to show that there exist values $z^*$ that satisfy the
payoff indifference condition
\begin{eqnarray}  P(A_t(\theta,z^*)-\theta>0|z^*)=c. \label{l0}
\end{eqnarray} In Appendix \ref{A3} we show that there exist small values $\check{z}$ such that
$P(A_t(\theta,\check{z})-\theta>0|\check{z})>c$ and large values
$\hat{z}$ such that $P(A_t(\theta,\hat{z})-\theta>0|\hat{z})<c$.

Equilibrium values $z^*$ are thus ensured: these values $z^*$ can
either obtain as solutions to (\ref{l0}), or the critical $z^*$
values are those values where the function
$P(A_t(\theta,z^*)-\theta>0|z^*)$ is discontinuous in $z^*$. In
that case we have a $z^*$, such that for small $\delta>0$,
$P(A_t(\theta,z^*+\delta)-\theta>0|z^*+\delta)>c$ and
$P(A_t(\theta,z^*-\delta)-\theta>0|z^*-\delta)<c$. That is, the
expected value of attacking/not attacking changes discontinuously
at $z^*$, and $z^*$ is the agents' equilibrium cutoff value.
\end{proof}

The current example, in particular equation (\ref{A2}), illustrates how runs feed on themselves: the size of the attack $A$ is an increasing function $\pi(A)$ of the attack itself. Put differently, increases in the mass of attacking agents increase the mass of agents who receive high signal realizations, and thus induce more agents to join the run. This increase in the number of attacking agents is once again visible, and induces even more agents to run... . Hence, if the private signal's precision is sufficiently high, runs feed on themselves, and agents can coordinate on multiple equilibria. Put differently, when bystanders see that the riot police is outnumbered, they are tempted to join in on the protest, which, in turn, attracts an even bigger turn-out... .





\section{Discussion}\label{Discussion}

Suppose you are passing by your local bank branch, and you see people lining up to withdraw money. Suppose also, that you have some knowledge of the bank's balance sheet and it's cash reserves "$\theta$". How would you choose your strategy? 

Following the global games logic, you know that your information over $\theta$ is correlated with the other agents' information over $\theta$. Accordingly, you can use your signal over the bank's balance sheet to infer the other agents' information over the bank's balance sheet. In turn, since all agents condition their actions on their information over the bank's finances, you can compute a posterior distribution over the other players' actions.  

Alternatively, you can use your information over the other players' actions: If you see people lining up in front of your local bank, a run is underway, and you should withdraw your money while you can. On the contrary, if nobody lines up to withdraw, then there is probably no run, and you do not have to act.

The current framework takes the middle ground between these two extremes: agents observe each others' actions, and they do think about the fundamental. The model is thus flexible enough to accommodate a range of environments that vary regarding the information that agents have over each others' actions and over fundamentals. 

In the context of bank-runs, currency crises, or riots and revolutions, agents will arguably monitor each others' actions closely, and the current model predicts that precise private information brings back multiplicity. Another application concerns models of aggregate economic activity,\footnote{In \citet{Dia82} type models, firms face a binary production choice, and the probability of making a sale increases with the level of overall economic activity $A$ and monetary policy $\theta$.} where firms' production choices depend on aggregate economic activity and exogenous fundamentals.

\newpage
\appendix


\section{Proof of Proposition \ref{infogap_mult_general}}\label{proof_infogap_mult_general}
To prove Proposition \ref{infogap_mult_general} we begin with a lemma, which collects a number of useful properties. In Section \ref{Iteration} we use these properties, together with the theorem of Arzelà–Ascoli, to show that there exists a continuum of equilibria $A_t(\theta)$. Regarding notation $A^n_t(t^-)$ and $A^n_t(t^+)$  denote left and right limits of the function $A_t(\theta)$ depicted in Diagram \ref{Diagram3}.

\begin{lemma} \label{helper_lemma_2d}
Suppose that for a given $t$ $A(\theta) \in [1-\delta,1]$ if $\theta < t$, and $A(\theta) \in [0,\delta]$ otherwise. Then the following holds:
\begin{enumerate}
\item If $y_i \geq 1-\delta-\gamma$ and $x_i\leq t+ \xi$, then $P[\theta<t|A(\cdot),(x_i,y_i)] \geq c$.
\item If $y_i \leq \delta+\gamma$ and $x\geq t- \xi$, then $P[\theta\geq t|A(\cdot),(x_i,y_i)] \geq c$.
\item $P[y_i\geq 1-\delta-\gamma \cap x\leq t+ \xi|A(\cdot),\theta < t]\geq 1-\delta$
\item $P[y_i\leq \delta+\gamma \cap x> t- \xi|A(\cdot),\theta \geq t]\geq 1-\delta$
\end{enumerate}
\end{lemma}
\begin{proof}
\begin{enumerate}
\item To prove the first statement, we define
$$\kappa_N := \frac{\int_{t}^{t+N} f_x(x_i-\theta)f_y(y_i-A(\theta))\, d\theta}{\int_{t-N}^t f_x(x_i-\theta)f_y(y_i-A(\theta))\, d\theta}\, ,$$
and show that the first two conditions can be reduced to statements about $\kappa_N$. We denote by $\mathcal{U}^t_N$ the uniform distribution of $\theta$ over the interval $[t-N,t+N]$.
\begin{multline*}
P[\theta<t|A(\cdot),(x_i,y_i)] =\lim_{N\rightarrow\infty} \frac{P[\theta<t \wedge (x_i,y_i)|A(\cdot),\mathcal{U}^t_N]}{P[(x_i,y_i)|A(\cdot),\mathcal{U}^t_N]} \\
= \lim_{N\rightarrow\infty} \frac{\int_{t-N}^t \frac1{2N}f_x(x_i-\theta)F_y(y_i-A(\theta)) \, d\theta}{\int_{t-N}^t \frac1{2N}f_x(x_i-\theta)f_y(y_i-A(\theta)) \, d\theta+\int_t^{t+N} \frac1{2N}f_x(x_i-\theta)f_y(y_i-A(\theta)) \, d\theta} \\
= \lim_{N\rightarrow\infty} \frac{1}{1+\kappa_N} \geq c \Leftrightarrow \lim_{N\rightarrow\infty} \kappa_N \leq \frac{1-c}{c}
\end{multline*}
Using appropriate variable transformations ($\tau = \theta - t$ in the nominator and $\tau = t-\theta$ in the denominator) and symmetry of $f_x$ we get
$$\kappa_N = \frac{\int_{0}^{N} f_x(x_i-(t+\tau))f_y(y_i-A(t+\tau))\, d\tau}{\int_{0}^{N} f_x(x_i-(t-\tau))f_y(y_i-A(t-\tau))\, d\tau} \, .$$
Recall, that $y_i \geq 1-\delta-\gamma$ and $x_i\leq t+ \xi$ holds in this case.
\begin{multline*}
\kappa_N \leq \frac{\int_{0}^{N} f_x(x_i-(t+\tau))\, d\tau}{\int_{0}^{N} f_x(x_i-(t-\tau))\, d\tau} \, \frac{\sup_{a\in[0,\delta]} f_y(y_i-a)}{\inf_{a\in[1-\delta,1]} f_y(y_i-a)} \\
= \frac{F_x(x_i-t)-F_x(x_i-t-N)}{F_x(x_i-t+N)-F_x(x_i-t)} \, \frac{\sup_{a\in[0,\delta]} f_y(y_i-a)}{\inf_{a\in[1-\delta,1]} f_y(y_i-a)} \\
\underset{{N\rightarrow \infty}}\longrightarrow \frac{F_x(x_i-t)}{1-F_x(x_i-t)} \, \frac{\sup_{a\in[0,\delta]} f_y(y_i-a)}{\inf_{a\in[1-\delta,1]} f_y(y_i-a)} \\
\leq \frac{F_x(\xi)}{1-F_x(\xi)} \, \frac{\sup_{a\in[0,\delta]} f_y(y_i-a)}{\inf_{a\in[1-\delta,1]} f_y(y_i-a)}\leq \frac{1-c}{c}
\end{multline*}
The last but one inequality uses condition \eqref{cond1_2d}.
\item The proof of the second statement relies on the same arguments used in 1. First, observe that
$$P[\theta\geq t|A(\cdot),(x_i,y_i)] = \lim_{N\rightarrow\infty} \frac{1}{1+\frac1{\kappa_N}} \geq c \Leftrightarrow \lim_{N\rightarrow\infty} \frac1{\kappa_N} \leq \frac{1-c}{c}\, .$$
Note that $y_i \leq \delta+\gamma$ and $x\geq t- \xi$ holds. Again, we can obtain a statement: 
\begin{multline*}
\frac1{\kappa_N} \leq \frac{1-F_x(-\xi)}{F_x(-\xi)}\, \frac{\sup_{a\in[1-\delta,1]} f_y(y_i-a)}{\inf_{a\in[0,\delta]} f_y(y_i-a)} \\
= \frac{F_x(\xi)}{1-F_x(\xi)}\, \frac{\sup_{a\in[0,\delta]} f_y((1-y_i)-a)}{\inf_{a\in[1-\delta,1]} f_y((1-y_i)-a)} \leq \frac{1-c}{c} \, .
\end{multline*}
The equality uses the symmetry of $F_x$ and $f_y$. The last inequality exploits condition \eqref{cond1_2d} (note that $(1-y_i) > 1 - \delta - \gamma$).
\item The third statement follows from condition \eqref{cond2_2d}. Regarding notation, we use $A^n_t(t^-)$ to denote left limits of the function $A_t$ depicted in Diagram \ref{Diagram3}:
\begin{multline}
P[y_i\geq 1-\delta-\gamma \wedge x\leq t+ \xi|A(\cdot),\theta < t] 
\\
\geq P[y_i\geq 1-\delta-\gamma \wedge x\leq t+ \xi|A(\theta)=1-\delta,\theta = t^{-}] \\
\geq F_x[t+\xi-t](1-F_y[1-\delta-\gamma-(1-\delta)]) = F_x[\xi]F_y[\gamma]
\geq 1-\delta
\end{multline} 
\item The fourth statement follows from condition \eqref{cond2_2d} and the following inequalities:
\begin{multline}
P[y_i\leq \delta+\gamma \wedge x> t- \xi|A(\cdot),\theta \geq t] \\
\geq P[y_i\leq \delta+\gamma \wedge x> t- \xi|A(\theta)=\delta,\theta = t] \\
\geq F_y[\gamma](1-F_x[-\xi]) = F_x[\xi]F_y[\gamma]
\geq 1-\delta
\end{multline}
\end{enumerate}
\end{proof}

\subsection{Iteration}\label{Iteration}

We now use an iteration to prove the existence of an equilibrium for a given\footnote{The requirement $1-\delta \geq \gamma$, $1>3\delta + 2\gamma$ in Proposition \ref{infogap_mult_general} implies that the interval is non-degenerated.}
\begin{equation}
t\in[\delta+\gamma,1-\delta-\gamma]\, . \label{def_t}
\end{equation}

We start from a hypothetical situation in which player $i$ faces an aggregate attack $A^0_t$ defined by (see green line below)
\begin{equation}\label{Attheta0}
A^0_t(\theta) := \left\{\begin{array}{ll}
1-\delta \qquad\qquad & \theta < t \\
\delta \qquad\qquad & \text{otherwise.}
\end{array} \right.
\end{equation}

We define the set of signals for which a player who would find it optimal to attack given a conjecture about other players' behavior $A^n_t$ by $\Gamma^n$. Using this definition we have
\begin{equation}
A^{n+1}_t(\theta) := \iint_{\Gamma^n} f_x(x-\theta)f_y(y-A_t^n(\theta))\,dx \, dy\,\label{xyz}
\end{equation} 

Given this definition of $A^{n+1}_t$, it follows by induction from the inequalities in Lemma \ref{helper_lemma_2d} that $A^n_t \geq 1-\delta$ ($\leq \delta$) for $\theta\leq t$ (otherwise) for all $n\geq 0$.

To complete the proof of Proposition \ref{infogap_mult_general} we use the theorem of Arzelà–Ascoli to show that the sequence $A_t^n$ has a convergent subsequence. That is, we note that $A_t^n$ are continuous, uniformly bounded (by zero and one), and have a uniformly bounded derivative for all $\theta \neq t$, which implies equicontinuity. Hence, the preconditions of the Arzelà–Ascoli theorem are met\footnote{See, e.g., \citeauthor{shilov2013elementary} (\citeyear{shilov2013elementary}, p. 32).} on each interval $[-k,t]$ (define $A_t^n$ at $t$ by the right limit) and $[t,k+1]$ for $k\in\mathbf{N}$ (they are also met for subsequences). Hence, for $k=1$ there is a convergent subsequence on $[-1,t]$ denoted by $A_t^{n_1}$, from which we can select yet another convergent subsequence on $[t,2]$ (and thus on $[-1,2]$) denoted by $A_t^{n_2}$. We can carry out the same procedure for each $k>1$, and receive a (sub)sequence $A_t^{n_{2k}}$ that converges on $[-k,k+1]$. Last but not least, select the $k$-th element of sequence $A_t^{n_{2k}}$ to create a new sequence $A_t^{n_0}$. This sequence is a subsequence of a converging sequence, and hence it converges.

We denote the limit of this sequence by $A_t$. It constitutes an equilibrium due to the continuity of the best-response operator.


\section{Proof of Proposition \ref{infogap_mult_1d}}\label{proof_infogap_mult_1d}

The proof is an adapted version of the proof of Proposition \ref{infogap_mult_general}. Again, we first state and proof a supporting lemma:

\begin{lemma} \label{helper_lemma_1d}
Suppose that for a given $t$ $A(\theta) \in [1-\delta,1]$ if $\theta < t$, and $A(\theta) \in [0,\delta]$ otherwise. Then the following holds:
\begin{enumerate}
\item If $z_i \geq 1 - \delta - t- \gamma$ then $P[\theta<t|A(\cdot),z_i] \geq c$.
\item If $z_i \leq \delta - t + \gamma$, then $P[\theta\geq t|A(\cdot),z_i] \geq c$.
\item $P[z_i \geq 1 - \delta - t- \gamma|A(\cdot),\theta < t]\geq 1-\delta$
\item $P[z_i \leq \delta - t + \gamma|A(\cdot),\theta \geq t]\geq 1-\delta$
\end{enumerate}
\end{lemma}
\begin{proof}
\begin{enumerate}
\item To prove the first inequality in 1. we define
$$\kappa_N := \frac{\int_{t}^{t+N} g(z_i-A(\theta)+\theta) \, d\theta}{\int_{t-N}^{t} g(z_i-A(\theta)+\theta) \, d\theta}\, ,$$
and show that both inequalities in 1. can be reduced to statements about $\kappa_N$.
\begin{multline*}
P[\theta<t|A(\cdot),z_i] =\lim_{N\rightarrow\infty} \frac{P[\theta<t \wedge z_i|A(\cdot),\mathcal{U}^t_N]}{P[z_i|A(\cdot),\mathcal{U}^t_N]} \\
= \lim_{N\rightarrow\infty} \frac{\int_{t}^{t+N} g(z_i-A(\theta)+\theta) \, d\theta}{\int_{t}^{t+N} g(z_i-A(\theta)+\theta) \, d\theta+ \int_{t}^{t+N} g(z_i-A(\theta)+\theta) \, d\theta} \\
= \lim_{N\rightarrow\infty} \frac{1}{1+\kappa_N} \geq c \Leftrightarrow \lim_{N\rightarrow\infty} \kappa_N \leq \frac{1-c}{c}
\end{multline*}
Due to continuity of $g$ there exist $\alpha_N\in[0,\delta]$ and $\beta_N\in[1-\delta,1]$ such that
\begin{multline*}\kappa_N = \frac{\int_{t}^{t+N} g(z_i-\alpha_N+\theta) \, d\theta}{\int_{t-N}^{t} g(z_i-\beta_N+\theta) \, d\theta} = \frac{G(z-\alpha_N+\theta+N)-G(z-\alpha_N+\theta)}{G(z-\beta_N+\theta)-G(z-\beta_N+\theta-N)} \\
\underset{{N\rightarrow \infty}}\longrightarrow \frac{1-G(z-\alpha+\theta)}{G(z-\beta+\theta)} \leq \frac{1-c}{c} \, ,
\end{multline*}
where $\alpha\in[0,\delta]$ and $\beta\in[1-\delta,1]$ are the limits of the respective sequence. 
The last inequality uses condition \eqref{cond1_infogap_mult_1d} and implies that $P[\theta<t|A(\cdot),z_i] \geq c$.
\item The proof of property 2. works along the lines of step 1., or alternatively like step 2. in the proof of Lemma \ref{helper_lemma_2d}.
\item To prove the third property, we use the notation $A^n_t(t^-)$ to denote left limits of the function $A_t$ depicted in Diagram \ref{Diagram3}:
\begin{multline*}
P[z_i \geq 1 - \delta - t - \gamma|A(\cdot),\theta < t] \geq  P[z_i \geq 1 - \delta - t - \gamma|A(\theta) = 1- \delta,\theta = t^-] \\
= P[\rho_i\geq - \gamma] = 1 - G(-\gamma) = G(\gamma) \geq 1 - \delta
\end{multline*}
The last but one equality holds due $G$ being symmetric, the last equality due to condition \eqref{cond2_infogap_mult_1d}.
\item Part 4. follows from:
\begin{multline*}
P[z_i \leq \delta - t + \gamma|A(\cdot),\theta \geq t] \geq  P[z_i \leq \delta - t + \gamma|A(\theta) = \delta,\theta =t] \\
= P[\rho_i\leq \gamma] = G(\gamma) \geq 1 - \delta
\end{multline*}
Again, the last equality is due to condition \eqref{cond2_infogap_mult_1d}.
\end{enumerate}
\end{proof}

The remaining part of the proof of Proposition \ref{infogap_mult_general} applies directly with two exceptions:
\begin{enumerate}
\item For given $A^n$ there exists a cutoff $z_n$ such that a player attacks iff $z_i<z_n$. This allows to adjust the definition of $A^{n+1}_t$ given $A^n_t$;. 
$$A^{n+1}_t(\theta) := \int_{-\infty}^{z_n} g (z- A^n_t(\theta)+\theta)) \, dz \, .$$
\item To establish equicontinuity of $A^n_t$, we start with another supporting Lemma:
\begin{lemma}\label{xn_lemma}
	Suppose agents consider the aggregate attack (function) to be equal to $A^n_t$ and that $A^n_t(\theta) > t + \delta $ for $\theta<t$, and $A^n_t(\theta) < t-\delta$ otherwise. Then, the agent's attack cutoff is $z^n\in(-\gamma,\gamma)$.
\end{lemma}
\begin{proof}
	An agent receiving a signal $z_i\geq \gamma$ would attack, where we use the notation $A^n_t(t^-)$ to denote left limits and $A^n_t(t^+)$ to denote right limits:
	\begin{multline*}
	P[\theta < t|z_i=z,A^n_t] =\footnotemark \frac{G(z-A^n_t(t^-)+t)}{G(z-A^n_t(t^-)+t) + 1 - {G(z-A^n_t(t^+)+t)}}\\
	\geq\footnotemark \frac{1-G(z+\delta)}{G(z-\delta) + 1 - G(z+\delta)} 
	= P[\theta < t|z_i=z,A^0_t]> P^*
	\end{multline*}
	Hence,\footnotetext{Use analogous computations as in step 2.} for the cutoff $z^n<\gamma$ has to hold. \footnotetext{The inequality can be shown by rearranging the following inequality: $$\frac{G(z-A^n_t(t^-)+t)}{G(z-\delta)}\geq 1 \geq \frac{1-G(z-A^n_t(t^+)+t)}{1-G(z+\delta)}$$} On the other hand, $z^n>-\gamma$ holds, as an agent receiving a signal $z_i\leq -\gamma$ would not attack:
	\begin{multline*}
	P[\theta\geq t|z_i=z,A^n_t] =  \frac{1 - G(z-A^n_t(t^+)+t)}{G(z-A^n_t(t^-)+t) + 1 - {G(z-A^n_t(t^+)+t)}}\\
	\geq \frac{1-G(z+\delta)}{G(z-\delta) + 1 - G(z+\delta)} = P[\theta\geq t|z_i=z,A^0_t]> 1-P^*
	\end{multline*}
\end{proof}

Now, note that $A_t^n$ are continuous, uniformly bounded and have a uniformly bounded derivative for all $\theta \neq t$:
{\small\begin{align*}\frac{d A_t^{n+1}}{d \theta}(\theta) &= - g(z^n-A_t^n(\theta)+\theta)\left(\frac{d A_t^n}{d \theta}(\theta)-1\right) \\
	&= -g(z^n-A_t^n(\theta)+\theta)\, \frac{1-g(z^n-A_t^n(\theta)+\theta)^n}{1-g(z^n-A_t^n(\theta)+\theta)}\\
	&\geq -\max\left\{g(\gamma-\delta)\, \frac{1-g(\gamma-\delta)^n}{1-g(\gamma-\delta)},g(-\gamma+\delta)\, \frac{1-g(-\gamma+\delta)^n}{1-g(-\gamma+\delta)} \right\}\\
	&> \frac{-g(\gamma-\delta)}{1-g(\gamma-\delta)}
	\end{align*}}The second equality can be shown by induction using $\frac{d A_t^0}{d \theta}(\theta) = 0$. We use the notation to $A_t^n(t^-)$ for the right limit of $A_t^n$ for $\theta\rightarrow t$. The first inequality holds because $g$ is increasing (decreasing) on $\theta < t$ (otherwise), $z^n\in(-\gamma,\gamma)$ (Lemma \ref{xn_lemma}), and $A_t^n(\theta)\geq t+ \delta$ on $\theta < t$ and $A_t^n(\theta) < t - \delta$ otherwise. The second inequality holds due to the symmetry of $g$ and condition \eqref{cond3_infogap_mult_1d}. Also note, that the second equality implies that the derivative of $A_t^n$ is non-positive. Thus, we have established uniform upper and lower bounds.
\end{enumerate}

\section{Conditional Probability}\label{A3}
We have to show that, $\lim_{z^*\rightarrow -\infty}
P(A(\theta,z^*)-\theta>0|z^*)=0$ and $\lim_{z^*\rightarrow \infty}
P(A(\theta,z^*)-\theta>0|z^*)=1$. To do so, we recall that
$z_i=A-\theta+\sigma_z\xi_i$ with $\xi_i\sim\mathcal{N}(0,1)$, and
examine the conditional probability:
\begin{eqnarray}  P(A(\theta,z^*)-\theta>0|z^*)=c. \label{l01} \end{eqnarray}
We begin by defining $y(\theta,z^*):=A(\theta,z^*)-\theta$, and we
recall Bayes's formula:
\begin{eqnarray}  f(y|z^*)=\frac{h(z^*|y)f(y)}{\int_{-\infty}^{\infty} h(z^*|y)f(y)dy}, \label{l02} \end{eqnarray}
Where $h(z^*|y)$ is a normal distribution. Moreover, we note that:
\begin{eqnarray} f(y(\theta))=g(\theta(y))\frac{d \theta}{d y}. \label{l06} \end{eqnarray}
Where $\frac{d\theta}{dy}=\frac{1}{A_{\theta}(\theta,z^*)-1}$.
Recalling (\ref{A2}), $A=1-\Phi(\sqrt{\alpha_z}(z^*-A+\theta))$,
we note that
$A_{\theta}(\theta,z^*)=\frac{-\sqrt{\alpha_z}\phi(\sqrt{\alpha_z}(z^*-A+\theta))}{1-\sqrt{\alpha_z}\phi(\sqrt{\alpha_z}(z^*-A+\theta))}<0$,
for $A_j, j=1,3$. Agents hold a uniform uninformative prior over
$\theta$, such that $g(\theta)$ is a constant $\bar{g}$. Moreover,
we have $\lim_{z^*\rightarrow \infty}A_{\theta}=0$ and thus
$\lim_{z^*\rightarrow \infty}\frac{d\theta}{dy}=-1$ and
$f(y)=-1\bar{g}$. Substituting into (\ref{l02}) yields:
\begin{eqnarray}  \lim_{z^*\rightarrow \infty}f(y|z^*)=\lim_{z^*\rightarrow \infty}\frac{h(z^*|y)}{\int_{-\infty}^{\infty} h(z^*|y)dy}, \label{l02a} \end{eqnarray}
where $h(z^*|y)=\phi(\sqrt{\alpha_z}(z^*-y))$. Finally, we have:
\begin{eqnarray}  \lim_{z^*\rightarrow \infty}P(y>0|z^*)=\lim_{z^*\rightarrow \infty}\frac{\int_{0}^{\infty}h(z^*|y)}{\int_{-\infty}^{\infty} h(z^*|y)dy}=\lim_{z^*\rightarrow \infty}\Phi(\alpha_z z^*)=1. \label{l02b} \end{eqnarray}
The same argument can be made to show that $\lim_{z^*\rightarrow
-\infty} P(y>0|z^*)=0$.

\newpage

\addcontentsline{toc}{section}{References}
\markboth{References}{References}
\bibliographystyle{apalike}
\bibliography{References}

\end{document}